\documentclass[10pt,twocolumn,twoside,submit]{JCNtran}

\usepackage{amsfonts}
\usepackage[dvips]{graphicx}
\usepackage{times}
\usepackage{cite}
\usepackage{amsmath}
\usepackage{array}
\usepackage{amssymb}

\usepackage{stfloats}
\usepackage{slashbox}
\usepackage{graphicx}
\usepackage{footnote}
\usepackage{booktabs}
\usepackage{array}
\usepackage{algorithmic}
\usepackage{algorithm}
\usepackage{subeqnarray}
\usepackage{cases}
\usepackage{threeparttable}
\usepackage{color}

\title{Spectrum Leasing and Cooperative Resource Allocation in Cognitive OFDMA Networks}
\author{Meixia Tao and Yuan Liu
\thanks{Manuscript received March 6, 2012; revised June 30, 2012; accepted September 2, 2012; approved for publication by Prof. Young-June Choi.}
\thanks{This work is supported by the Innovation Program of Shanghai Municipal Education Commission under grant 11ZZ19, and the Program for New Century Excellent Talents in University (NCET) under grant NCET-11-0331.}
\thanks{The authors are with the Department of
Electronic Engineering at Shanghai Jiao Tong University, Shanghai,
200240, P. R. China. Emails: \{mxtao, yuanliu\}@sjtu.edu.cn.}
}

\begin{document}
\maketitle

\newtheorem{lemma}{Lemma}
\newtheorem{remark}{Remark}
\newtheorem{proposition}{Proposition}

\vspace{-1.2cm}
\begin{abstract}
This paper considers a cooperative OFDMA-based cognitive radio network where the primary system leases some of its subchannels to the secondary system for a fraction of time in exchange for the secondary users (SUs) assisting the transmission of primary users (PUs) as relays. Our aim is to determine the cooperation strategies among the primary and secondary systems so as to maximize the sum-rate of SUs while maintaining quality-of-service (QoS) requirements of PUs. We formulate a joint optimization problem of PU transmission mode selection, SU (or relay) selection, subcarrier assignment, power control, and time allocation. By applying dual method, this mixed integer programming problem is decomposed into parallel per-subcarrier subproblems, with each determining the cooperation strategy between one PU and one SU. We show that, on each leased subcarrier, the optimal strategy is to let a SU exclusively act as a relay or transmit for itself. This result is fundamentally different from the conventional spectrum leasing in single-channel systems where a SU must transmit a fraction of time for itself if it helps the PU's transmission. We then propose a subgradient-based algorithm to find the asymptotically optimal solution to the primal problem in polynomial time.  Simulation results demonstrate that the proposed algorithm can significantly enhance the network performance.
\end{abstract}


\section{Introduction}
\setlength\arraycolsep{2pt}

Cognitive radio (CR), with its ability to sense unused frequency bands and
adaptively adjust transmission parameters, has recently attracted
considerable interest for solving the spectrum scarcity problem
\cite{Haykin,Akyildiz}. A key concept in cognitive radio networks
(CRNs) is opportunistic or dynamic spectrum access, which allows
secondary users (SUs) to opportunistically access the bands licensed to
primary users (PUs).
%
%
Most of the works on dynamic spectrum access regard the secondary
transmission as harmful interference and hence the SUs do not
participate in the primary transmission.
%
%
%
Recently, a new cooperation strategy between the primary system and the secondary system was proposed in\cite{Simeone}
and further investigated in \cite{Zhang}. Therein, the PU link leases its channel to the SUs for a
fraction of time to transmit secondary traffic in exchange for the SUs acting as relays to assist the transmission of primary traffic.
%
The spectrum-leasing based cooperation can improve the performance
of both the primary and secondary systems and result in a ``win-win"
situation.

The early works  \cite{Simeone} and \cite{Zhang} on spectrum leasing only investigated the time slot allocation in the
single-PU, multi-SU, and single-channel scenario.
%
More specifically, in \cite{Simeone}, one PU targets at maximizing
its rate while multiple SUs compete with each other to access the
single channel. However, this scheme may result in an extreme case
that the PU is aggressive and the SUs have no opportunity to access
the channel. Recall that in CRNs, PUs are willing to share the
spectrum resource with SUs if their quality-of-service (QoS)
requirements are satisfied \cite{Haykin,Akyildiz}. In \cite{Zhang},
the PU maximizes its utility in terms of rate and revenue while the
SUs competitively make decisions based on their rates and payments.
Nevertheless, the virtual payment and revenue may lead to another
extreme case that the PU provides all of the transmission time to
the SUs on the single channel, which is not practical in CRNs.

In this paper, we consider the general spectrum leasing and resource
allocation problem in multi-channel multi-user CRN based on
orthogonal frequency-division multiple-access (OFDMA). The
motivation of using OFDMA is two-fold. First, OFDMA is not only
adopted in many current and next generation wireless standards but
also a strong candidate for CRNs \cite{Weiss}. The second is that
OFDMA-based systems can flexibly incorporate dynamic resource
allocations in CRNs (e.g., \cite{Wang,Bansal,Ma}). The primary
system consists of multiple user pairs conducting bidirectional
communication. The secondary system is a cellular network consisting
of a base station (BS) and a set of SUs.
The two systems operate in a cognitive and cooperative manner by
allowing the SUs to occupy certain subcarriers given that the QoS of
the PUs are satisfied with the assistance of SUs as cooperative
relays.

As CRNs are typically hierarchical and heterogeneous, it is
intuitive that if SUs can aggressively help PUs' transmission, then
less subcarriers will be needed by the PUs to satisfy their QoS
requirements, and as a result the SUs can access more subcarriers
for maximizing their own data rates.
%
%
Meanwhile, as the communication in the primary system is
bidirectional, the cooperation of SU as relays can also bring
network coding gain in the form of two-way relaying\footnote{In two-way
relay systems, a pair of nodes exchange information with the help
of a relay node using physical layer network coding \cite{Rankov,Popovski,Kim}. Two-way relaying can achieve much higher spectral
efficiency than the traditional one-way relaying.}. Thus,
more subcarriers can be leased to the SUs. The increased
spectral efficiency is in turn transformed into cooperation
opportunities.
Optimizing such cooperative CRN has unique attractiveness and
challenges as follows.

Firstly, for the primary system, when relaying is necessary, it has
to decide which cooperative transmission modes (one-way relaying and
two-way relaying) to select and which set of SUs to choose, since it
has higher priority in a CRN.
Secondly, for the secondary system, it needs to schedule appropriate
SUs to utilize the leased subcarriers for maximizing its total
throughput. Moreover, for those SUs that not only be selected as
relays but also be scheduled to transmit for themselves, the
secondary system needs to balance their resource utilization.
Thirdly, from the common perspective of the primary and secondary
systems, it is crucial to determine which set of subcarriers to
cooperate on together with how much power and time slots to transmit
signals, in order to satisfy the QoS requirements of the
primary system.

The main contributions of this paper are summarized as follows:
\begin{enumerate}
\item We propose an optimization framework for joint bidirectional transmission mode selection, SU selection, subcarrier assignment, power control, and time slot allocation in
the cooperative CRNs. The objective is to maximize the sum-rate of
all SUs while satisfying the individual rate requirement for each of
the PUs.
There are three distinct features in our optimization framework.
\emph{First}, by subcarrier assignment and allocating time slot
between PUs and SUs in cooperation sessions, multiuser diversity can
be achieved in both frequency domain and time domain. \emph{Second},
as the communication in the primary system is bidirectional, we can
exploit network coding gain in the form of two-way relaying to
improve spectral efficiency via the SUs' assistance. \emph{Third},
using the OFDMA-based relaying architecture, each PU pair can
conduct the bidirectional communication by multiple transmission
modes, namely direct transmission, one- and two-way relaying, each
of them can take place on a different set of subcarriers.

\item We show that in the multi-channel cooperative CRNs, the optimal strategy is to let
a SU exclusively act as a relay for a PU or transmit data for
itself on a cooperated channel. This result fundamentally differs
from the conventional cooperation in the single-channel scenario
where a SU must transmit a fraction of time for itself if it
forwards the PU's transmission on the channel.

\item Using the Lagrange dual decomposition method, the joint
optimization problem is decomposed into parallel
per-subcarrier-based subproblems. An efficient algorithm is proposed
to find the asymptotically optimal solution in polynomial time.
\end{enumerate}

The remainder of this paper is organized as follows. Section II
introduces the optimization framework, including system model and
problem formulation. Section III presents the details of the
Lagrange dual decomposition method for the joint resource-allocation
problem. Section IV provides the simulation results. Finally, we
conclude this paper in Section V.

\section{Optimization Framework}

We consider an OFDMA-based CRN where the primary system coexists
with the secondary system as shown in Fig.~\ref{fig:system}. The
primary system is an ad hoc network, consisting of multiple user pairs with each user pair conducting bidirectional communications.
The secondary system of
interest is the uplink of a single-cell network where a BS
communicates a set of SUs. Note that the downlink can be analyzed in
the same way.
The proposed model can be justified in the IEEE 802.22 standard,
where the CR systems are based on cellular basis.
By taking advantage of the parallel OFDMA-based relaying
architecture, each PU pair can conduct the bidirectional
communication through three transmission modes, namely direct
transmission, one- and two-way relaying, on different sets of
subcarriers. As shown in Fig.~\ref{fig:slots}, the PUs can transmit
directly and the SUs can access the PUs' residual subcarriers, or
they transmit by a cooperation manner. On each cooperated
subcarrier, a SU can assist a PU (or PU pair) using one- or two-way
relaying.
This setup can fully explore available diversities of the network,
including user, channel, and transmission mode.

\begin{figure}[t]
\begin{centering}
\includegraphics[scale=.5]{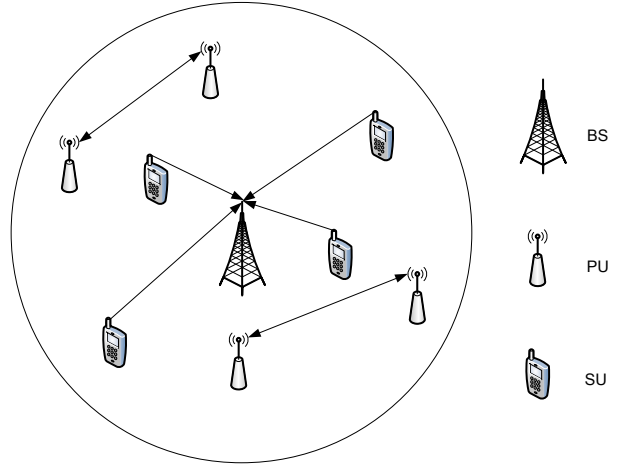}
\vspace{-0.1cm}
 \caption{System architecture of the CRN.}\label{fig:system}
\end{centering}
\vspace{-0.3cm}
\end{figure}
\begin{figure}[t]
\begin{centering}
\includegraphics[scale=.6]{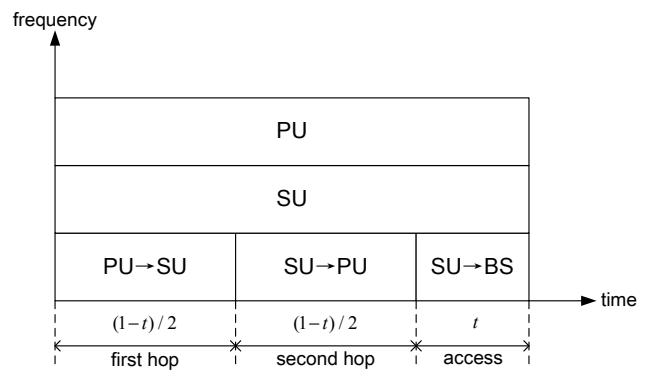}
\vspace{-0.1cm}
 \caption{Time slot allocation between a PU and a SU. PU and SU can transmit directly (on different subcarriers) or by a cooperation manner.}\label{fig:slots}
\end{centering}
\vspace{-0.3cm}
\end{figure}

We model the wireless fading environment by large-scale path loss
and shadowing, along with small-scale frequency-selective fading.
The channels between different links experience independent fading
and the network operates in slow fading environment, so that channel
estimation is perfect.
We assume that the two-hop transmission uses the same subcarrier for
both links, i.e., the source$\rightarrow$relay link and the
relay$\rightarrow$destination link.
%
The time slot allocation between a PU and a SU on a cooperated
subcarrier is illustrated in Fig.~\ref{fig:slots}, in which we
further assume that the two hops of the cooperative transmission use
equal time slots. This is true for amplify-and-forward (AF) relaying
strategy because AF needs equal time allocation, but more
flexibility can be provided if the two hops pursue time adaptation
for decode-and-forward (DF). Nevertheless, we still adopt equal time
slot allocation between the two hops for simplicity.

Let $ {\cal{N}}=\{1,2,\cdot\cdot\cdot,N\}$ denote the set of
subcarriers and ${\cal{K}}=\{1, \cdots, k, \cdots, K\}$ denote the
set of users, with the first $K_P$ being the PU pairs and the
remaining $K_S=K-K_P$ being SUs. Here $k$ represents PU pair index
if $1\leq k\leq K_P$ and represents SU index if $K_P+1\leq k\leq K$.
Denote $k_1$ and $k_2$ as the two users in the $k$-th PU pair,
$1\leq k\leq K_P$.
Denote
$\boldsymbol P_n=[P_{1,n},\cdot\cdot\cdot,P_{k,n},\cdot\cdot\cdot,P_{K,n}]^T$
and
$\boldsymbol{R}_n=[R_{1,n},\cdot\cdot\cdot,R_{k,n},\cdot\cdot\cdot,R_{K,n}]^T$
as the power and achievable rate vectors on subcarrier $n$,
respectively. If $1\leq k\leq K_P$,
$P_{k,n}=[P_{k_1,n},P_{k_2,n}]^T$ and
$R_{k,n}=[R_{k_1,n},R_{k_2,n}]^T$.
For interference avoidance, at most one PU (or PU pair) and one SU
are active on each subcarrier.
%
This exclusive subcarrier assignment and best relay selection can be
implicitly involved in $\boldsymbol P_n$ and $\boldsymbol{R}_n$.
Let $\boldsymbol{P^{{\rm max}}}=[P_1^{{\rm
max}},\cdot\cdot\cdot,P_k^{{\rm max}},\cdot\cdot\cdot,P_K^{{\rm
max}}]^T$ denote the peak power constraints vector. Again, if $1\leq
k\leq K_P$, $P_k^{{\rm max}}=[P_{k_1}^{{\rm max}},P_{k_2}^{{\rm
max}}]^T$. Let
$\boldsymbol{r}=[r_1,\cdot\cdot\cdot,r_k,\cdot\cdot\cdot,r_{K_P}]^T$
(with $r_k=[r_{k_1},r_{k_2}]^T$) be the rate requirements of the
PUs. Denote
$\boldsymbol{t}_n=[t_{K_P+1,n},\cdot\cdot\cdot,t_{k,n},\cdot\cdot\cdot,t_{K,n}]^T$
whose element $0\leq t_{k,n}\leq 1$ is the duration that SU $k$
transmits on subcarrier $n$. Note that as aforementioned there is at
most one SU active on a subcarrier, thus at most one non-zero
element in $\boldsymbol{t}_n$.
%
%
Without loss of generality, we assume that additive white noises at
all nodes are independent circular symmetric complex Gaussian random
variables, each of them has zero mean and unit variance.
Assuming channel reciprocity in time-division duplex, we then use
$|h_{k_1,k_2,n}^p|^2$, $|h_{k',{\rm BS},n}^s|^2$, and
$|h_{k_1,k',n}^{ps}|^2$ ($|h_{k_2,k',n}^{ps}|^2$) to represent the
effective channel gains between PU $k_1$ and $k_2$ of PU pair $k$,
SU $k'$ and BS, and PU $k_1$ ($k_2$) and SU $k'$, respectively, on
subcarrier $n$. For brevity, we denote all of them as a vector
$\boldsymbol{H}_n$.
A PU can cooperate with multiple SUs and a SU can assist multiple
PUs. Thanks to the use of OFDMA, the inter-user interference can be
avoided. In addition, the intra-pair interference for the PU pairs
will be treated as back-propagated self-interference and canceled
perfectly after two-way relaying. Finally, we let
$\boldsymbol P=[\boldsymbol P_1,\cdot\cdot\cdot,\boldsymbol P_n,\cdot\cdot\cdot,\boldsymbol P_N]^T$,
$\boldsymbol{R}=[\boldsymbol{R}_1,\cdot\cdot\cdot,\boldsymbol{R}_n,\cdot\cdot\cdot,\boldsymbol{R}_N]^T$,
and
$\boldsymbol{t}=[\boldsymbol{t}_1,\cdot\cdot\cdot,\boldsymbol{t}_n,\cdot\cdot\cdot,\boldsymbol{t}_N]^T$
be the power, achievable rate, and time slot allocation matrices,
respectively.

In this paper, our objective is not only to optimally allocate
power, subcarriers, and time slot but also to choose best
transmission modes and relays for the PUs so as to maximize
the sum-rate of all SUs while satisfying the individual rate
requirement for each of the PUs. Mathematically, the optimization
problem can be formulated as

\begin{subequations}\label{eqn:joint}
\begin{align}
\max_{\boldsymbol P,\boldsymbol{R},\boldsymbol{t}}&\sum_{k=K_P+1}^K
\sum_{n=1}^N R_{k,n}\\
{\rm s.t.}~~& \sum_{n=1}^NP_{k,n}\leq P_k^{{\rm max}},~~\forall k \label{eqn:peak_power}\\
& \sum_{n=1}^N R_{k,n}\geq r_k,~~1\leq k\leq K_P \label{eqn:qos}\\
& \boldsymbol P\succeq 0,~\boldsymbol{t}\in[0,1] \label{eqn:positive}\\
& R_{k,n}\in R \left(\boldsymbol P_n,\boldsymbol{t}_n,\boldsymbol{H}_n\right),~1\leq
k\leq K_P,\forall n\\
& R_{k,n}=t_{k,n}C\left(\frac{P_{k,n}|h_{k,{\rm
BS},n}^s|^2}{\sigma_{\rm
BS}^2}\right),K_P+1\leq k\leq K,\forall n,\label{eqn:region}
\end{align}
\end{subequations}
where $C(x)=\log_2(1+x)$, $\sigma_{\rm BS}^2$ is the noise variance at the BS, and $  {R}$ is the set of achievable rates for the PUs,
which is related to $\boldsymbol P_n$, $\boldsymbol{t}_n$,
$\boldsymbol{H}_n$, and the transmission modes. Note again that the
exclusive subcarrier assignment and best relay selection are
implicitly involved in $\boldsymbol P_n$, $\boldsymbol{R}_n$, and
$\boldsymbol{t}_n$.

\begin{remark}
In this paper, we assume that a central controller is available, so
that the network channel state information and sensing results can
be reliably gathered for centralized processing. Notice that the
centralized CRNs are valid in IEEE 802.22 standard \cite{Stevenson},
where the cognitive systems operate on a cellular basis and the
central controller can be embedded with a base station (BS).
This assumption is also reasonable if a spectrum broker
exists in CRNs for managing
spectrum leasing and access \cite{Brik,Jia}.
Such centralized approach is commonly used in a variety of CRNs
(e.g.,
\cite{Sun,Hossain,Hoang,Zhao,Brik,Jia,Wang,Bansal,Ma,Kang}). Compared with distributed approaches (e.g., \cite{Simeone,Zhang,Yi}), a CRN having a central manager that
possesses detailed information about the wireless network enables
highly efficient network configuration and better enforcement of a
complex set of policies \cite{Brik}.
\end{remark}

\begin{remark}
For CRNs, there is no single figure of QoS merit to measure the
performance of the primary system. In this paper, we choose the rate
requirement as the QoS metric. Other QoS metrics, like outage
probability and signal-to-interference-plus-noise ratio (SINR), can
be easily accommodated in the problem formulation.
Moreover, the weighted sum-rate maximization for the SUs can be
taken into account for the fairness issue, which does not affect the
proposed algorithms in the sequel.
\end{remark}

\section{Lagrange Dual Decomposition Based Optimization}

The optimization problem in (\ref{eqn:joint}) is a mixed integer
programming problem. In \cite{YuTCOM}, the authors show that for the
nonconvex resource optimization problems in OFDMA systems,  the
duality gap becomes zero under the time-sharing condition. It is
also proved in\cite{YuTCOM} that the time-sharing condition is
always satisfied as the number of OFDM subcarriers goes to infinity,
regardless of the nonconvexity of the original problem. This means
that solving the original problem and solving its dual problem are
 equivalent.
Based on the result, the Lagrange dual decomposition method is
recently applied to OFDMA-based cellular and cognitive radio
networks in \cite{Yu} and \cite{Kang}, respectively.
%
In this section, we shall apply the result from \cite{YuTCOM} to
solve our problem in (\ref{eqn:joint}). In particular, some valuable
insights are obtained for multi-channel cooperative CRNs, which
shows that the generalization from the single-channel
case\cite{Simeone,Zhang} to the multi-channel case is nontrivial.
%
%
%
%

%
We first introduce two sets of dual variables,
$\boldsymbol{\lambda}=[\lambda_1,\cdot\cdot\cdot,\lambda_k,\cdot\cdot\cdot,\lambda_K]^T$
($\lambda_k=[\lambda_{k_1},\lambda_{k_2}]^T$ if $1\leq k\leq K_P$)
and
$\boldsymbol{\beta}=[\beta_1,\cdot\cdot\cdot,\beta_k,\cdot\cdot\cdot,\beta_{K_P}]^T$
($\beta_k=[\beta_{k_1},\beta_{k_2}]^T$) associated with constraints
(\ref{eqn:peak_power}) and (\ref{eqn:qos}) respectively, where
$\boldsymbol{\lambda}\succeq 0$ and $\boldsymbol{\beta}\succeq 0$.
The Lagrange of the problem in (\ref{eqn:joint}) can be written as

\begin{eqnarray}
&&L(\boldsymbol P,\boldsymbol{R},\boldsymbol{t},\boldsymbol{\lambda},\boldsymbol{\beta})=\sum_{k=K_P+1}^K\sum_{n=1}^N
R_{k,n}\nonumber\\ &&+ \sum_{k=1}^K\lambda_k\left(P_k^{{\rm
max}}-\sum_{n=1}^NP_{k,n}\right) +
\sum_{k=1}^{K_P}\beta_k\left(\sum_{n=1}^N R_{k,n}- r_k\right).
\end{eqnarray}
Define $  {D}$ as the set of all primary variables
$\{\boldsymbol P,\boldsymbol{R},\boldsymbol{t}\}$ that satisfy
constraints (\ref{eqn:positive})-(\ref{eqn:region}). The dual
function is given by

\begin{equation}\label{eqn:dualf}
g(\boldsymbol{\lambda},\boldsymbol{\beta})=\max_{\{\boldsymbol P,\boldsymbol{R},\boldsymbol{t}\}\in
{D}}L(\boldsymbol P,\boldsymbol{R},\boldsymbol{t},\boldsymbol{\lambda},\boldsymbol{\beta}),
\end{equation}
and the dual optimization problem can be expressed as
\begin{subequations}\label{eqn:dual}
\begin{align}
\min_{\boldsymbol{\lambda},\boldsymbol{\beta}} ~~&
g(\boldsymbol{\lambda},\boldsymbol{\beta}) \\
{\rm s.t.} ~~~&
\boldsymbol{\lambda}\succeq0,\boldsymbol{\beta}\succeq0.
\end{align}
\end{subequations}

The dual function (\ref{eqn:dualf}) can be rewritten as
\begin{equation}
g(\boldsymbol{\lambda},\boldsymbol{\beta})=\sum_{n=1}^N
g_n(\boldsymbol{\lambda},\boldsymbol{\beta}) + \sum_{k=1}^K\lambda_k
P_k^{{\rm max}} - \sum_{k=1}^{K_P}\beta_k r_k,
\end{equation}
where
\begin{eqnarray}\label{eqn:per}
g_n(\boldsymbol{\lambda},\boldsymbol{\beta})=\max_{\{\boldsymbol P,\boldsymbol{R},\boldsymbol{t}\}\in
{D}}\Bigg[\sum_{k=K_P+1}^KR_{k,n} + \sum_{k=1}^{K_P}\beta_kR_{k,n}\nonumber\\ -
\sum_{k=1}^K\lambda_kP_{k,n}\Bigg]
\end{eqnarray}
are the $N$ independent per-subcarrier-based optimization
subproblems.

Since a dual function is always convex by definition\cite{Boyd},
subgradient-based ellipsoid method \cite{ellipsoid} can be used to
minimize $g(\boldsymbol{\lambda},\boldsymbol{\beta})$ by updating
$\{\boldsymbol{\lambda}, \boldsymbol{\beta}\}$ simultaneously along
with appropriate search directions, and it is guaranteed to converge
to the optimal solution $\{\boldsymbol{\lambda}^*, \boldsymbol{\beta}^*\}$.

\begin{proposition}\label{pro:gradients}
For the dual problem defined in (\ref{eqn:dual}),

\begin{equation}
\triangle\lambda_k=P_k^{{\rm max}}-\sum_{n=1}^NP_{k,n}^*,~1\leq k
\leq K,
\end{equation}
and
\begin{equation}
\triangle\beta_k=\sum_{n=1}^N R_{k,n}^*- r_k,~1\leq k \leq K_P,
\end{equation}
are subgradients of $g(\boldsymbol{\lambda},\boldsymbol{\beta})$,
where $\{P_{k,n}^*,R_{k,n}^*\}$ are the optimal solutions of
(\ref{eqn:per}) for given $\{\boldsymbol{\lambda},
\boldsymbol{\beta}\}$.
\end{proposition}

\begin{proof}
Please see Appendix~\ref{app:gradients}.
\end{proof}

%

As mentioned earlier, there are at most one PU (or PU pair), denoted
as $  P$, and one SU, denoted as $  S$, active on
a subcarrier. Here $  P$ and $  S$ also represent
the \emph{best} PU (or PU pair) and SU respectively, among all
possible users, that maximizes (\ref{eqn:per}) for a given
subcarrier $n$. This can be obtained by an exhaustive search. The
complexity is detailed later. Specifically, it needs to first
compute the optimal powers and rates for all users under all
transmission modes, then let one PU and/or SU under one transmission
mode that maximizes (\ref{eqn:per}) active on each subcarrier.
Therefore, the per-subcarrier problems in (\ref{eqn:per}) can be
alternatively expressed as

\begin{subequations}\label{eqn:per'}
\begin{align}
\max_{\boldsymbol P_n,\boldsymbol{R}_n,\boldsymbol{t}_n}~~&
R_{  S,n} + \beta_  PR_{  P,n}
-\lambda_{  P}P_{  P,n} - \lambda_{
S}P_{  S,n}\\
{\rm s.t.}~~~~~& P_{  P,n}\geq0,P_{  S,n}\geq0,0\leq t_{  S,n}\leq1\\
& R_{  P,n}\in  {R}\left(P_{
P,n},P_{  S,n},t_{
S,n},\boldsymbol{H}_n\right)\\
& R_{  S,n}=t_{  S,n}C\left(\frac{P_{
S,n}|h_{  S,{\rm BS},n}^s|^2}{\sigma_{\rm BS}^2}\right).
\end{align}
\end{subequations}

In what follows, for brevity of notation, the subscript $n$ in
(\ref{eqn:per'}) is omitted due to all $N$ per-subcarrier-based
subproblems having an identical structure. In addition, for direct
transmission and one-way relaying, it is observed that the
per-subcarrier optimization problem for the two links of a PU pair,
i.e., $  P_1\rightarrow  P_2$ and $
P_2\rightarrow  P_1$ (with or without relaying), has the
same structure and can be decoupled. Thus, for brevity, we only
consider here the $  P_1\rightarrow  P_2$ link as
an example. In fact, the per-subcarrier optimization for direct
transmission and one-way relaying needs to first compute the optimal
values of the objective function in (\ref{eqn:per'}) for both links
and then let one of them that has the maximum value active on the
subcarrier. Moreover, we let $\gamma=|h_{  P_1,
P_2,n}^p|^2$, $\gamma_1=|h_{  P_1,  S,n}^p|^2$,
$\gamma_2=|h_{  S,  P_2,n}^{ps}|^2$, and
$\gamma_s=|h_{  S,{\rm BS},n}^s|^2$.

\subsection{Direct Transmission}

In this transmission mode, either a PU or a SU occupies solely the given subcarrier. The per-subcarrier
optimization problem in (\ref{eqn:per'}) can be expressed as
\begin{subequations}\label{eqn:dir}
\begin{align}
\max_{P_{  P_1}\geq0,P_{  S}\geq0}~~& R_{
S}+\beta_{  P_2}R_{  P_2} -\lambda_{
P_1}P_{  P_1} -\lambda_{  S}P_{  S} \\
{\rm s.t.}~~~~~~&R_{  S}=C(P_{  S}\gamma_s)\\
&R_{  P_2}=C(P_{  P_1}\gamma).
\end{align}
\end{subequations}
Since the problem in (\ref{eqn:dir}) is
convex, by applying the Karush-Kuhn-Tucker (KKT) conditions
\cite{Boyd}, the optimal power allocations can be obtained as
\begin{equation}\label{eqn:dir1}
P_{  P_1}^*=\left(\frac{\beta_{  P_2}}{a
\lambda_{  P_1}} - \frac{1}{\gamma}\right)^+,
\end{equation}
and
\begin{equation}\label{eqn:s}
P_{  S}^*=\left(\frac{1}{a \lambda_{  S}} -
\frac{1}{\gamma_s}\right)^+,
\end{equation}
where $a=\ln2$ and $(x)^+=\max(0,x)$. For a given subcarrier, the
direct transmission further needs to compute the optimal values of
the objective function in (\ref{eqn:dir}) over one of $P_{
P_1}^*$ and $P_{  S}^*$ with the other being zero, and
then let one of them that has maximum value active. (\ref{eqn:dir1})
and (\ref{eqn:s}) show that the optimal power allocations are
achieved by multi-level water-filling. In particular, the water
level of each PU depends explicitly on its QoS requirement, and can
differ from one another. On the other hand, the water levels of all
SUs are the same.

\subsection{One-Way Relaying}

If relaying is required on a given subcarrier, a fraction of $1-t_  S$ time is
used by a PU to transmit the primary traffic with the help of a SU,
while the rest $t_  S$ time is assigned to the SU to transmit its own data.
In this paper, we focus on DF only for simplicity of presentation,
for both one-way relaying and two-way relaying.
Other relaying strategies are readily applicable to our
framework and algorithms. The detailed discussion is given later.

In case of DF one-way relaying, the per-subcarrier problem in
(\ref{eqn:per'}) can be rewritten as
\begin{subequations}\label{eqn:df}
\begin{align}
\max_{P_{  P_1},P_{  S},t_  S}~~&
R_{  S} + \beta_{  P_2}R_{  P_2}
-\lambda_{  P_1}P_{  P_1} - \lambda_{
S}P_{  S}\label{eqn:dfobj}\\
{\rm s.t.}~~~~~& P_{  P_1}\geq0, P_{  S}\geq0,
0\leq t_  S\leq1\\
& R_  S=t_  SC(P_  S\gamma_s)\label{eqn:factor1}\\
& R_{  P_2}=\frac{(1-t_
S)}{2}\min\left\{C(P_{
P_1}\gamma_1),C(P_{  P_1}\gamma + P_
S\gamma_2)\right\}.\label{eqn:factor2}
\end{align}
\end{subequations}

In (\ref{eqn:factor2}), the first term in the min-operation is the
achievable rate of the $  P_1\rightarrow  S$ link,
and the second term is the achievable rate by maximum ratio
combining between the $  S\rightarrow  P_2$ link
and $  P_1\rightarrow  P_2$ link.
The following proposition is established for determining the optimal
value of the time slot allocation variable $t_  S$.

\begin{proposition}\label{pro:time}
For each \emph{cooperated} subcarrier, a SU \emph{exclusively} acts
as a relay for cooperative transmission or transmits traffic for
itself.
\end{proposition}

\begin{proof}
The proposition means that the time slot allocation variable
$t_  S$ is binary, i.e., $t_  S^*\in\{0,1\}$,
which can be proved by contradiction.

Assume that the optimal solution of (\ref{eqn:df}) is $(t_
S^*,P_{  P_1}^*,P_  S^*)$ with $0<t_
S^*<1$. Next, we show that we can always find another better solution with $t_  S$ being binary.

Let us rewrite the objective function (\ref{eqn:dfobj}) as
\begin{eqnarray}
&&f(t_  S^*,P_{  P_1}^*,P_{  S}^*)=t_
S\Bigg[C(P_{  S}^*\gamma_s)\nonumber\\
&&-\beta_{P_2}\min\left\{\frac{1}{2}C(P_{
P_1}^*\gamma_1),\frac{1}{2}C(P_{  P_1}^*\gamma + P_{
S}^*\gamma_2)\right\}\Bigg]\nonumber\\
&& +\beta_{  P_2}\min\left\{\frac{1}{2}C(P_{
P_1}^*\gamma_1),\frac{1}{2}C(P_{  P_1}^*\gamma + P_{
S}^*\gamma_2)\right\}\nonumber\\
&&-\lambda_{  P_1}P_{  P_1}^* -
\lambda_{  S}P_{  S}^*.
\end{eqnarray}
%
%
%
If $C(P_  S^*\gamma_s)<\beta_{
P_2}\min\left\{\frac{1}{2}C(P_{
P_1}^*\gamma_1),\frac{1}{2}C(P_{  P_1}^*\gamma +
P_  S^*\gamma_2)\right\}$, we have
\begin{eqnarray}
 f(t_  S^*,P_{  P_1}^*,P_  S^*) &<& \beta_{
P_2}\min\left\{\frac{1}{2}C(P_{
P_1}^*\gamma_1),\frac{1}{2}C(P_{  P_1}^*\gamma +
P_  S^*\gamma_2)\right\}\nonumber\\
&-&\lambda_{
P_1}P_{  P_1}^* - \lambda_{  S}P_{
S}^*\nonumber\\
&=& f(0,P_{  P_1}^*,P_  S^*).
\end{eqnarray}
Similarly, if $C(P_
S^*\gamma_s)\geq\frac{\beta_{
P_2}}{2}\min\left\{C(P_{
P_1}^*\gamma_1),C(P_{  P_1}^*\gamma +
P_  S^*\gamma_2)\right\}$, we have
\begin{eqnarray}
f(t_  S^*,P_{  P_1}^*,P_  S^*) &\le&
C(P_{ S}^*\gamma_s^*)-\lambda_{
P_1}P_{  P_1}^* -\lambda_{  S}P_{  S}^* \nonumber \\
&<& C(P_{ S}^*\gamma_s^*)-\lambda_{  S}P_{  S}^* \nonumber \\
&=&f(1,0,P_  S^*).
\end{eqnarray}
These results contradict the assumption. This completes the proof.
\end{proof}

%
This proposition also holds for two-way relaying as discussed in the next subsection. The proof is similar and hence ignored.

Proposition~\ref{pro:time} significantly simplifies the
per-subcarrier optimization problem in (\ref{eqn:df}) \emph{without
loss of optimality} by an exhaustive search over $t_  S$.
Specifically, we set $t_  S=0$ and $t_  S=1$ to
compute the optimal values of (\ref{eqn:dfobj}), respectively, then
follow the one that has the maximum value.

The intuition is that, on a cooperated subcarrier (see
Fig.~\ref{fig:slots}), if the subcarrier condition on the
SU$\rightarrow$BS link is good but on the cooperative link is poor,
it is better that the PU leases the whole transmission time slot to
the SU. Otherwise, the SU completely devotes itself as a relay to
the PU. In other words, if a SU exclusively forwards a PU's traffic
on a subcarrier, the PU shall lease other subcarrier(s) to the SU as
remuneration. This \emph{channel-swap} based multi-channel
cooperation fundamentally
differs from the single-channel cooperation case
\cite{Simeone,Zhang} where if a SU forwards a PU's traffic, it must
benefit from the PU on the channel. The spectral efficiency
improvement brings more cooperation opportunities and leased
subcarriers, and thus, the total throughput of the secondary system
is increased.
In the following we consider $t_  S=0$ and $t_
S=1$, respectively.

Case 1: $t_  S=0$. In this case, $
S$ exclusively acts as a relay on a cooperated subcarrier. In DF
one-way relaying, it is intuitive that $R_{  P_2}$ is
maximized when $C(P_{  P_1}\gamma_1)=C(P_{
P_1}\gamma + P_  S\gamma_2)$, which leads to
\begin{equation}\label{eqn:ps}
P_  S=\gamma'P_{  P_1},
\end{equation}
where $\gamma'=(\gamma_1-\gamma)/\gamma_2$. It is noted that DF
occurs only if $\gamma_1> \gamma$.
Substituting (\ref{eqn:ps}) to (\ref{eqn:df}) and let $t_
S=0$, the problem can be rewritten as
\begin{subequations}
\begin{align}
\max_{P_{  P_1}\geq0}~~&\beta_{  P_2}R_{
P_2} - (\lambda_{
P_1}+\lambda_  S\gamma')P_{  P_1}\\
{\rm s.t.} ~~~&R_{  P_2}=\frac{1}{2}C(P_{
P_1}\gamma_1).
\end{align}
\end{subequations}
The above is a convex problem. By applying the KKT conditions, the
optimal power allocation is given by
\begin{eqnarray}\label{eqn:df1}
P_{  P_1}^*= \left[\frac{\beta_{  P_2}}{2a
(\lambda_{  P_1}+\lambda_  S\gamma')} -
\frac{1}{\gamma_1}\right]^+,
\end{eqnarray}
$P_  S^*$ can be obtained according to (\ref{eqn:ps}). The
above optimal power allocation (\ref{eqn:df1}) shows that higher
channel gain $\gamma_1$, meaning higher $\gamma'$, results in lower
water level, which is the extra feature compared with the standard
water-filling approach (e.g., (\ref{eqn:dir1}) and (\ref{eqn:s}) in
direct transmission). One also observes that lower channel gain
$\gamma_2$ leads to lower water level and vice versa.

Case 2: $t_  S=1$. In this case, $  S$ uses a cooperated subcarrier solely for
its own transmission.
%
The optimal power allocation can be easily obtained and is the same
as (\ref{eqn:s}).

\subsection{Two-Way Relaying}

The two-way communication between $  P_1$ and $
P_2$ assisted by $  S$ takes place in two phases.
Specifically, in the first phase, also known as multiple-access
(MAC) phase, $  P_1$ and $  P_2$ concurrently
transmit signals to the assisting $  S$. In the second
phase, known as broadcast (BC) phase, $  S$ broadcasts the
processed signals to both $  P_1$ and $  P_2$.
Different from direct transmission and one-way relaying, two-way
relaying must occur in pair
\cite{YuanTWC, Jitvanichphaibool, YuanTCOM,YuanWCL12,YuanGC12, YuanJSAC12}. Thus both the
bidirectional links are taken into account together for two-way
relaying.
Here we only analyze the case of $t_  S=0$, and the case of
$t_  S=1$ is omitted since the optimal power allocation is
the same as (\ref{eqn:s}) if $t_  S=1$.

The per-subcarrier problem in (\ref{eqn:per'}) can be expressed as
(recall that $t_  S=0$)
\begin{subequations}\label{eqn:xor}
\begin{align}
\max_{P_  P\geq0,P_  S\geq0,R_  P}
~~&\beta_{  P_1}R_{  P_1} + \beta_{
P_2}R_{  P_2}\nonumber\\
& - \lambda_{  P_1}P_{
P_1} - \lambda_{
P_2}P_{  P_2}-\lambda_  SP_  S\\
{\rm s.t.}~~~~~~~& R_  P\in   {R}\left(P_
P,P_  S,\gamma_1,\gamma_2\right)=
{C}_{\textrm{MAC}}\bigcap  {C}_{\textrm{BC}},
\end{align}
\end{subequations}
where $  {C}_{\textrm{MAC}}$ and $  {C}_{\textrm{BC}}$
are the capacity regions for the MAC and BC phases, respectively
\cite{Xie,Boche,Kim}. Specifically,

\begin{eqnarray}\label{eqn:mac}
  {C}_{\textrm{MAC}}&=&\Big\{[R_{  P_1}~R_{
P_2}]\Big| R_{  P_1}\leq \frac{1}{2}C(P_{
P_2}\gamma_2),R_{  P_2}\leq \frac{1}{2}C(P_{
P_1}\gamma_1),\nonumber
\\ && ~~~~~~~~~~~~~R_{P_1}+R_{  P_2}\leq \frac{1}{2}C(P_{
P_2}\gamma_2+P_{  P_1}\gamma_1)\Big\},
\end{eqnarray}
and
\begin{equation}\label{eqn:bc}
  {C}_{\textrm{BC}}=\Big\{[R_{  P_1}~R_{
P_2}]\Big| R_{  P_1}\leq \frac{1}{2}C(P_
S\gamma_1),R_{  P_2}\leq \frac{1}{2}C(P_
S\gamma_2)\Big\}.
\end{equation}
Note that the channel reciprocity is used in the BC phase, which is
justified by the time-division duplex mode. Since both $
{C}_{\textrm{MAC}}$ and $  {C}_{\textrm{BC}}$ are convex
sets, and so is their intersection, the problem in (\ref{eqn:xor})
is a convex problem and can be solved by convex techniques.

Let $\alpha_1$ and $\alpha_2$ be the two dual variables associated
with the two rate constraints in (\ref{eqn:bc}). We first
incorporate the two rate constraints in (\ref{eqn:bc}) into the
objective function and rewrite the Lagrange dual problem of
(\ref{eqn:xor}) as
\begin{subequations}\label{eqn:xor'}
\begin{align}
\min_{\alpha_1,\alpha_2}\max_{\{P_  P,R_
P\}\in  {C}_{\textrm{DF}}} & \beta_{
P_1}R_{  P_1} + \beta_{  P_2}R_{  P_2}\nonumber\\
&- \lambda_{  P_1}P_{  P_1} - \lambda_{
P_2}P_{  P_2}-\lambda_  SP_  S \nonumber\\
&-\alpha_1\left[R_{  P_1}-\frac{1}{2}C(P_
S\gamma_1)\right] \nonumber\\
&-\alpha_2\left[R_{
P_2}-\frac{1}{2}C(P_  S\gamma_2)\right]\\
{\rm s.t.} ~~~~~~~~~~& \alpha_1\geq0,\alpha_2\geq0,
\end{align}
\end{subequations}
where $  {C}_{\textrm{DF}}$ is the set of the remaining
constraints in (\ref{eqn:xor}) that $\{P_  P,R_
P\}$ must satisfy. The minimization over $\{\alpha_1,\alpha_2\}$
can be done using ellipsoid method with the fact that
$\frac{1}{2}C(P_  S\gamma_1)-R_{  P_1}$ and
$\frac{1}{2}C(P_  S\gamma_2)-R_{  P_2}$ are
subgradients of $\alpha_1$ and $\alpha_2$, respectively.
It is observed that the optimization variables in (\ref{eqn:xor'})
are separable. Therefore, the maximization over $\{P_
P,R_  P\}$ in (\ref{eqn:xor'}) can be decomposed into two
subproblems that can be solved separately. The two subproblems are
\begin{subequations}\label{eqn:sub1}
\begin{align}
\max_{P_  P\geq0,R_  P} & (\beta_{
P_1}-\alpha_1)R_{  P_1} + (\beta_{
P_2}-\alpha_2)R_{  P_2} \nonumber\\
& - \lambda_{
P_1}P_{  P_1} - \lambda_{  P_2}P_{
P_2}\\
{\rm s.t.}~~~& \{P_  P,R_  P\}\in
{C}_{\textrm{MAC}},
\end{align}
\end{subequations}
and
\begin{equation}\label{eqn:sub2}
\max_{P_  S\geq0}~~\frac{\alpha_1}{2}C(P_
S\gamma_1) + \frac{\alpha_2}{2}C(P_  S\gamma_2) -
\lambda_  SP_  S.
\end{equation}
For brevity of notation, we let $\alpha_1'=\beta_{
P_1}-\alpha_1$ and $\alpha_2'=\beta_{  P_2}-\alpha_2$. It
is noted that both $\alpha_1'$ and $\alpha_2'$ must be nonnegative,
i.e., $\beta_{  P_1}\geq\alpha_1$ and $\beta_{
P_2}\geq\alpha_2$. In the following we present the solution to
each subproblem.

The subproblem in (\ref{eqn:sub1}) is a classic resource allocation
problem in the Gaussian MAC \cite{Tse}, where the optimal power and
rate allocations can be achieved by successive decoding.
Specifically, users' signals are decoded one by one in an increasing
rate weight order \cite{Tse}. Without loss of generality, we assume
that $\alpha_1'\geq\alpha_2'$ (here $\alpha_1'$ and $\alpha_2'$ can
be regarded as the rate weights for $  P_1$ and $
P_2$, respectively). Based on the polymatroid structure of the
Gaussian MAC \cite{Tse}, we then incorporate the three rate
constraints in (\ref{eqn:mac}) into the objective function of
(\ref{eqn:sub1}), the subproblem in (\ref{eqn:sub1}) can be
expressed as
%
\begin{eqnarray}\label{eqn:sub1'}
\max_{P_  P\geq0}~~&&\frac{\alpha_1'}{2}C(P_{
P_2}\gamma_2)+\frac{\alpha_2'}{2}\left[C(P_{
P_2}\gamma_2+P_{  P_1}\gamma_1)-C(P_{
P_2}\gamma_2)\right]\nonumber\\
&&-\lambda_{  P_1}P_{
P_1}-\lambda_{  P_2}P_{  P_2}.
\end{eqnarray}
It is easy to validate that the objective function of
(\ref{eqn:sub1'}) is jointly concave in $P_{  P_1}$ and
$P_{  P_2}$. By applying the KKT conditions, the optimal
power allocations can be obtained as
\begin{equation}\label{eqn:tw-p2}
P_{
P_2}^*=\left[\frac{(\alpha_1'-\alpha_2')\gamma_1}{2a(\gamma_1\lambda_{
P_2}-\gamma_2\lambda_{  P_1})} -
\frac{1}{\gamma_2}\right]^+,
\end{equation}
and
\begin{equation}\label{eqn:tw-p1}
P_{
P_1}^*=\frac{1}{2a}\left[\frac{\alpha_2'}{\lambda_{
P_1}} -
\frac{(\alpha_1'-\alpha_2')\gamma_2}{\gamma_1\lambda_{
P_2}-\gamma_2\lambda_{  P_1}}\right]^+.
\end{equation}
It is observed that the optimal power allocations in the MAC phase
have the form of water-filling. Moreover, the following proposition
is provided according to the above closed-form power allocations.

\begin{proposition}\label{pro:twoway}
For $\alpha_1'\geq\alpha_2'$, a necessary condition for the
occurrence of two-way relaying is $\gamma_1\lambda_{
P_2}>\gamma_2\lambda_{  P_1}$.
\end{proposition}

\begin{proof}
Please see Appendix~\ref{app:twoway}.
\end{proof}

The subproblem in (\ref{eqn:sub2}) is also convex since its
objective function is concave in $P_  S$. By applying the
KKT conditions, the optimal power allocation is given by
\begin{eqnarray}\label{eqn:tw-ps}
P_
S^*=\begin{cases}\frac{-\theta_2+\sqrt{\theta_2^2-4\theta_1\theta_3}}{2\theta_1},~&\textrm{if}~\lambda_  S<\frac{\alpha_1\gamma_1+\alpha_2\gamma_2}{2a}\\
0,~&\textrm{otherwise}\end{cases},
\end{eqnarray}
where $\theta_1=2a\lambda_  S\gamma_1\gamma_2$,
$\theta_2=2a\lambda_
S(\gamma_1+\gamma_2)-\gamma_1\gamma_2(\alpha_1+\alpha_2)$, and
$\theta_3=2a\lambda_  S-\alpha_1\gamma_1-\alpha_2\gamma_2$.

%

\begin{remark}
For both one- and two-way relaying, direct transmission mode is
optimal if $P_{  S}^*=0$. In this case, the optimal power
and rate allocations are the same as those obtained in the direct
transmission mode in Section III-A. Moreover, it is noted that
two-way relaying occurs if $P_  S^*$, $P_{
P_1}^*$, and $P_{  P_2}^*$ are all positive. For two-way
relaying, another interesting case is that $P_{  S}^*$ is
positive and one of $P_{  P_1}^*$ and $P_{
P_2}^*$ is equal to zero. In other words, one direction is
inactive. In this case, one-way relaying must be optimal.
\end{remark}

\begin{remark}
  In our centralized framework, the cooperation is between the primary system
  and secondary system rather than among individual users. Thus some SUs may not
   transmit their own traffic because they are the best
  option for primary traffic relaying. In this case, these SUs may
  be re-scheduled for transmission at the next transmission frame by a higher layer scheduler for
  long-term fairness. However, analysis on higher layer scheduling
  is beyond of the scope of this paper.
\end{remark}

%

\begin{remark}
In this paper, we employ DF just for an illustration purpose. Other
relaying strategies, like AF and compress-and-forward (CF), are
generally applicable to our proposed framework and algorithms.
However, the achievable rate expressions of AF and CF are not
concave, for both one- and two-way relaying.\footnote{It is noted
that the nonconvexity does not affect Proposition~\ref{pro:time}.}
To overcome the difficulty, some approximations can be adopted for
AF and CF such that the achievable rate expressions are concave, and
thus they can be solved using convex optimization techniques.
%
%
\end{remark}

After obtaining the optimal solution in the dual domain, we now need
to obtain the optimal solution to the original primal problem in
(\ref{eqn:joint}). Due to the non-zero duality gap, the optimal
solution obtained in the dual domain may not satisfy all the
constraints in the original primal problem. To tackle this problem,
we first obtain the optimal transmission mode selection and
user-assignment for each subcarrier using the method in the dual
domain, then the primal problem in (\ref{eqn:joint}) reduces to a
pure power allocation problem and it is convex. By applying KKT
conditions, the optimal power allocations follow the same
expressions in the dual domain and the details are omitted here.
This approach is \emph{asymptotically} optimal due to the vanishing
duality gap when the number of subcarriers is sufficiently large
\cite{YuTCOM}.

At the end of the section, we analyze the computational complexity
of the proposed algorithm.
The complexity of determining $  P$ and $  S$ on
each subcarrier for direct transmission is $  {O}(2K_P+K_S)$,
and for one- and two-way relaying are $  {O}(2K_PK_S)$ and
$  {O}(K_PK_S)$, respectively. Note that the complexity of
the search over $t_  S=1$ is implicitly contained in the
optimization of direct transmission mode. Therefore, the total
complexity of solving all $N$ per-subcarrier problems is $
{O}\left(N(2K_P+K_S+3K_PK_S)\right)$. Combining the complexity of
the ellipsoid method, the total complexity of solving the dual
problem is $
{O}\left(N(2K_P+K_S+3K_PK_S)(4K_P+K_S)^2\right)$, which is linear in
$N$ and polynomial in $K_P$ and $K_S$.

\section{Simulation Results}

In this section, we evaluate the performance of the proposed
cooperative scheme using simulation.
The conventional scheme without cooperation is
selected as a benchmark, which corresponds to the optimization of
direct transmission in our proposed algorithm and its complexity is $
{O}\left(N(2K_P+K_S)(4K_P+K_S)^2\right)$. As
another benchmark, the performance of the Fixed Transmission Mode
(FTM) based allocation is also presented.
%
%
In particular, the FTM scheme lets the transmission mode for each PU
be pre-fixed according to nodes' geographical information, and other
optimizations are the same with the proposed optimal algorithm. This
is attractive for practical systems where path loss dominates the
performance of the network nodes. In specific, a PU is assigned to
the direct transmission mode if the path loss (or distance) of the
source-destination link is smaller than that of all source-relay
links and the cooperative transmission modes are used otherwise.
When it is assigned the cooperative transmission modes, two-way
relaying is adopted if the path losses of the source-relay link and
the relay-destination link is about the same, otherwise one-way
relaying is used (in this case, $  P_1\rightarrow
P_2$ direction is performed). Note that if a PU is assigned two-way relaying, the other
PU in the same pair is also assigned two-way relaying.
%
%
For those PUs who need SUs' assistance, we assign the nearest SU to
each PU, the search over the suitable SU for each PU is reduced to
$  {O}(1)$. Hence, the total complexity of this suboptimal
algorithm is $  {O}\left(N(5K_P+K_S)(4K_P+K_S)^2\right)$.

We consider an a primary network in a 1 km by 1 km square area and a
cellular secondary network whose BS is located in the center of the
square and with $1$ km radius. All users are randomly but uniformly
distributed. The statistical path loss model and shadowing are
referred to \cite{Erceg}, where we set the path loss exponent to be
$4$ and the standard deviation of log-normal shadowing to be $5.8$
dB. The small-scale fading is modeled by Rayleigh fading process,
where the power delay profile is exponentially decaying with maximum
delay spread of $5$ $\mu s$. A total of $2000$ independent channel
realizations were used. Different channel realizations are with
different node locations.
We set the number of OFDM subcarriers be $N=64$.  Without loss of
generality, we let all users have the same maximum power
constraints, and all PUs have the same rate requirements. In all of
the simulations, we fix $K_P=2$ PU pairs (or equivalently $4$ PUs)
in the network. Without loss of generality, we let all users have
the same peak power constraints.

\begin{figure}[t]
\begin{centering}
\includegraphics[scale=.6]{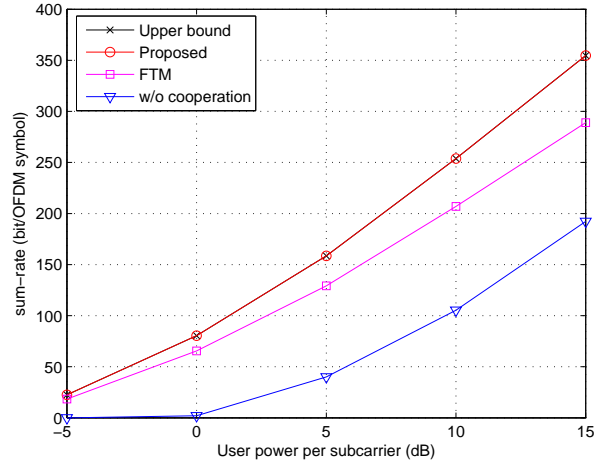}
\vspace{-0.1cm}
 \caption{Sum-rate of the secondary system versus transmit SNR per subcarrier, with $K_P=2$ PU pairs, $K_S=4$ SUs and rate requirement $5$ bit/OFDM symbol for all PUs.}\label{fig:4KS}
\end{centering}
\vspace{-0.3cm}
\end{figure}
\begin{figure}[t]
\begin{centering}
\includegraphics[scale=.6]{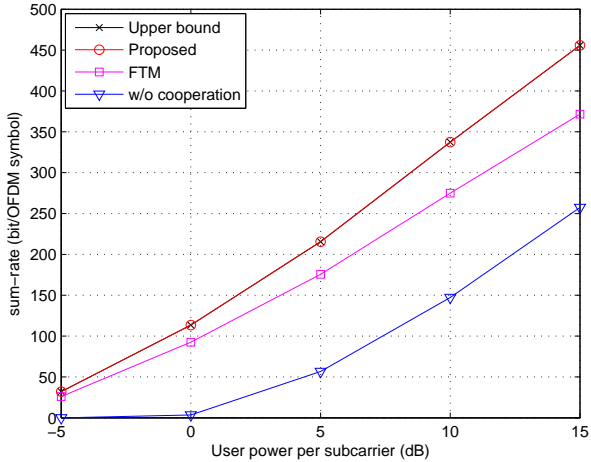}
\vspace{-0.1cm}
 \caption{Sum-rate of the secondary system versus peak power per subcarrier, with $K_P=2$ PU pairs, $K_S=8$ SUs and rate requirement $5$ bit/OFDM symbol for all PUs.}\label{fig:8KS}
\end{centering}
\vspace{-0.3cm}
\end{figure}

Figs.~\ref{fig:4KS} and \ref{fig:8KS} compare the sum-rate of the
secondary system versus user peak power (in dB) achieved by
different schemes when there are $K_S=4$ and $K_S=8$ SUs,
respectively. In both figures, the PU rate requirement is $5$
bits/OFDM symbol for each PU, and the dual optimum values serve as
the performance upper bounds. It is first observed that the proposed
cooperation scheme approaches the upper bound very tightly, which
verifies the effectiveness of the proposed algorithm.
One also observes that the proposed scheme outperforms the
conventional non-cooperative scheme by a significant margin. In
particular, compared with the conventional scheme, about $60\%$
throughput improvement is achieved in our proposed scheme. The
tremendous improvement is as the remuneration for cooperative
diversity, selection diversity, and network coding gain that the
secondary systems provides to the primary system. Second, one also
observes that our proposed scheme improves $20\%$ throughput over
the FTM scheme. This clearly suggests the benefits of bidirectional
transmission mode adaptation and SU selection for the PUs. Third,
from Fig.~\ref{fig:8KS} with $K_S=8$ SUs, it is observed that our
proposed scheme also outperforms the FTM and conventional schemes
substantially. Note that a larger $K_S$ results in higher
computational complexity mainly due to the updates of dual
variables.

\begin{figure}[t]
\begin{centering}
\includegraphics[scale=.6]{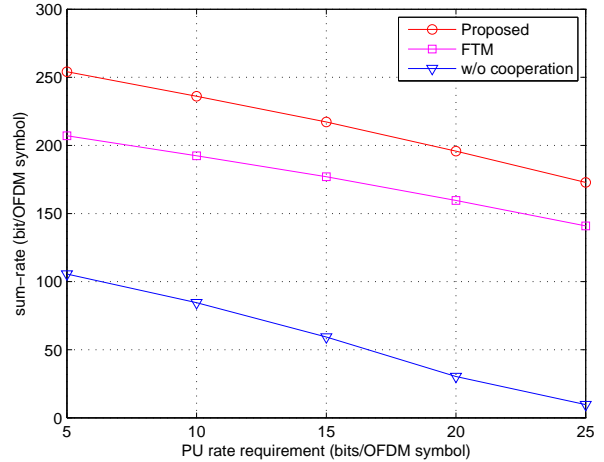}
\vspace{-0.1cm}
 \caption{Sum-rate of the secondary system versus PU rate requirement, with $K_P=2$ PU pairs, $K_S=4$ SUs and peak power is $10$ dB per subcarrier.}\label{fig:qos}
\end{centering}
\vspace{-0.3cm}
\end{figure}

We next study the sum-rate of the secondary system versus the
different PU rate requirements in Fig.~\ref{fig:qos}, where we fix
transmit SNR $10$ dB and $K_S=4$ SUs. As expected, our proposed
scheme outperforms the FTM and conventional schemes considerably
over all ranges of PU rate requirements. This further demonstrates
the effectiveness of the proposed scheme.

\section{Conclusion}

This paper studied the OFDMA-based bidirectional CRNs with
cooperation between the primary and secondary systems, for
supporting communication services with diverse QoS requirements. We
proposed an optimization framework for joint optimization of
bidirectional transmission mode selection, SU selection, subcarrier
assignment, power control, and time slot allocation. We converted
this mix integer programming problem with exponential complexity
into a convex problem using the dual decomposition method and
developed efficient algorithms with polynomial complexity.

A few important conclusions have been made throughout this paper.
%
Firstly, the time slot allocation between a PU and a SU on a
cooperated subcarrier is binary. Secondly, the proposed framework
can greatly improve the total throughput of the secondary system by
about $60\%$, compared with the non-cooperative scheme. Thirdly,
choosing the appropriate transmission modes for the PUs is
necessary. Last but not least, transmission mode adaptation and SU
selection over different subcarriers can enhance the total
throughput by about $20\%$.

The proposed algorithm can be used as the performance upper bound for
suboptimal or distributed algorithms. In future work, it will be interesting to investigate incentive-based
distributed schemes.

\appendices
%

\section{Proof of Proposition~\ref{pro:gradients}}\label{app:gradients}
By definition of $g(\boldsymbol{\lambda},\boldsymbol{\beta})$ in
(\ref{eqn:dualf}), we have
\begin{eqnarray}
g(\boldsymbol{\lambda}',\boldsymbol{\beta}')&\geq&\sum_{k=K_P+1}^K\sum_{n=1}^N
R_{k,n}^* + \sum_{k=1}^K\lambda_k'\left(P_k^{{\rm
max}}-\sum_{n=1}^NP_{k,n}^*\right) \nonumber\\
&+&\sum_{k=1}^{K_P}\beta_k'\left(\sum_{n=1}^N R_{k,n}^*- r_k\right)
\nonumber\\
&=&g(\boldsymbol{\lambda},\boldsymbol{\beta})+\sum_{k=1}^K(\lambda_k'-\lambda_k)\left(P_k^{{\rm
max}}-\sum_{n=1}^NP_{k,n}^*\right)\nonumber\\
 &+&\sum_{k=1}^{K_P}(\beta_k'-\beta_k)\left(\sum_{n=1}^N R_{k,n}^*-
r_k\right).
\end{eqnarray}
Hence, Proposition~\ref{pro:gradients} is proven by using the
definition of subgradient.

\section{Proof of Proposition~\ref{pro:twoway}}\label{app:twoway}
Both $P_{  P_2}^*$ and $P_{  P_2}^*$ must be
positive if two-way relaying occurs, besides $P_  S^*$ is
positive. We first investigate $P_{  P_2}^*$ in
(\ref{eqn:tw-p2}). It is easy to observe that
$\gamma_1\lambda_{  P_2}$ must be greater than
$\gamma_2\lambda_{  P_1}$, otherwise the first term in
(\ref{eqn:tw-p2}) is negative, and thus, $P_{  P_2}^*=0$.
For $P_{  P_2}^*$ in (\ref{eqn:tw-p1}), we let the first
term is greater than that of the second term, i.e.,
$\alpha_2'/\lambda_{
P_1}>(\alpha_1'-\alpha_2')\gamma_2/(\gamma_1\lambda_{
P_2}-\gamma_2\lambda_{  P_1})$. After some manipulations,
we obtain $\gamma_1\lambda_{
P_2}/\gamma_2\lambda_{  P_1}>\alpha_1'/\alpha_2'$.
Combining the condition $\alpha_1'/\alpha_2'\geq1$, we obtain
$\gamma_1\lambda_{  P_2}>\gamma_2\lambda_{  P_1}$.
This completes the proof.

\bibliographystyle{IEEEtran}
\bibliography{IEEEabrv,CR}

\end{document}